\documentclass{toc}

\tocdetails{
title = {On Solving Reachability in Grid Digraphs Using a Pseudoseparator},
  number_of_pages = {21}, 
  number_of_bibitems = {16},
  number_of_figures = {8},
  conference_version = {FSTTCS19}, %
  author = {Rahul Jain and Raghunath Tewari},
  authorlist = {Rahul Jain, Raghunath Tewari},
  acmclassification = {G.2.2},
  amsclassification = {68Q25},
  keywords = {%
digraph reachability, grid digraph, %
graph algorithm, sublinear space algorithm},
}
\tocdetails{%
volume = 19,      %
number = 2,       %
year = 2023,
received = {January 13, 2020},
revised = {May 25, 2022},
published = {August 28, 2023},
doi = {10.4086/toc.2023.v019.a002},
}   %

\usepackage{packages}
\usepackage{graphicx}   %

\usepackage{aumacros}

\begin{document}

\begin{frontmatter}%

\title{On Solving Reachability in Grid Digraphs  %
using a Pseudoseparator\titlefootnote{A conference version of this paper appeared in
the \href{https://doi.org/10.4230/LIPIcs.FSTTCS.2019.19}{Proceedings of the 39th IARCS Annual Conference on Foundations of Software Technology and Theoretical Computer Science (FSTTCS 2019)}~\cite{conf-version}.}}

\author[Jain]{Rahul Jain \orcidlink{0000-0002-8567-9475}}
\author[Tewari]{Raghunath Tewari}

\begin{abstract}
The reachability problem %
asks to decide %
if there exists a path from one vertex to another in a digraph. 
In a grid digraph, the vertices are
the %
points of a two-dimensional square grid,  %
and an edge can occur between a vertex and its 
immediate horizontal %
and vertical neighbors only.

Asano and Doerr (CCCG'11) presented the first simultaneous time-space
bound for reachability in grid 
digraphs %
by
solving
the problem in polynomial time and $O(n^{1/2 + \epsilon})$ space. In 2018, the space complexity was improved to $\tilde{O}(n^{1/3})$ by Ashida and Nakagawa (SoCG'18).

In this paper, we show that there exists a polynomial-time algorithm that uses $O(n^{1/4 + \epsilon})$ space to solve the reachability problem in a grid 
digraph %
containing $n$ vertices. We define and construct a new separator-like device called pseudoseparator to develop a divide-and-conquer algorithm. This algorithm works in a space-efficient manner to solve reachability.
\end{abstract}

%

\end{frontmatter}

\section{Introduction}
Let $s$ and $t$ be vertices of a given 
directed graph (digraph). %
The problem of 
digraph %
reachability is to decide if there is a path from $s$ to $t$. This problem has %
many  %
applications in the field of algorithms and computational complexity theory. Many algorithms for network-related problems use it as a subroutine. 
Digraph reachability is \textsf{NL}-complete and thus
captures  %
the complexity of %
nondeterministic  %
logarithmic space. Hence designing better algorithms for this problem is of utmost importance
to the theory of computing.  %

Standard 
traversal algorithms such as DFS and BFS give a
linear-time   %
algorithm for this problem, but they require linear space. Savitch's 
divide-and-conquer 
algorithm %
solves  %
reachability in $O(\log^2 n)$ space. However, as a tradeoff, it requires $n^{O(\log n)}$ time \cite{Savitch}. Hence it is natural to ask whether we can get the best of both worlds and design an algorithm for 
digraph %
reachability that runs in polynomial time and uses \emph{polylogarithmic} space. Wigderson asked a relaxed version of this question in his survey, whether 
digraph %
reachability can be solved by a polynomial-time algorithm that uses $O(n^{1-\epsilon})$ space \cite{Wigderson}.

Barnes et al.  %
showed that 
digraph reachability %
can be decided simultaneously in $n/2^{\Theta(\sqrt{\log n})}$ space and polynomial time \cite{BBRS}. Although this algorithm gives a sublinear %
space bound,  %
it still does not %
answer Wigderson's question.

Undirected graphs can be considered as a specific type of digraphs,
where %
adjacency
is a symmetric relation.
Reachability in undirected %
graphs  %
can be solved in logspace \cite{Reingold}. %
Also, for certain classes of digraphs, where the underlying undirected graph is topologically restricted, we know of polynomial time algorithms with space complexity better than linear. %
Imai et al. presented  
a polynomial-time algorithm that %
solves  %
reachability for planar 
digraphs %
using $O(n^{1/2 + \epsilon})$ space for any $\epsilon >0$ \cite{Imai}. Later, Asano et al. presented a polynomial-time algorithm whose space complexity was $\tilde{O}(n^{1/2})$ for the same problem \cite{Asano}. For 
digraphs %
of higher genus, Chakraborty et al. presented a polynomial-time algorithm that uses $\tilde{O}(n^{2/3}g^{1/3})$ space. Their algorithm additionally requires an embedding of the 
underlying undirected %
graph on a surface of genus $g$, as part of the input \cite{ChakrabortyPavan}. They also gave an $\tilde{O}(n^{2/3})$-space %
algorithm for %
$H$-minor-free   %
digraphs %
which requires
a  %
tree decomposition of the 
underlying undirected %
graph as 
part of the %
 input and $O(n^{1/2 + \epsilon})$-space
algorithms for
$K_{3,3}$-minor-free digraphs and for $K_5$-minor-free digraphs. %

For layered planar 
digraphs %
Chakraborty and Tewari showed that for every $\epsilon >0$, there is an $O(n^{\epsilon})$-space %
algorithm \cite{Chakraborty}. Stolee and Vinodchandran presented a polynomial-time algorithm that, for any $\epsilon > 0$ solves \emph{reachability} in 
an acyclic digraph %
with $O(n^{\epsilon})$ sources and embedded on the surface of genus $O(n^\epsilon)$ using $O(n^{\epsilon})$ space \cite{Stolee}. For unique-path 
digraphs %
Kannan et al. presented an %
$O(n^{\epsilon})$-space, %
polynomial-time algorithm \cite{kannan}.

In this paper, grid 
digraphs %
are our concern. Grid
digraphs %
are a subclass of planar
digraphs. %
In a grid 
digraph, %
the  %
vertices are
the %
points %
of an $m \times m$ grid.  %
The edges can only occur between a vertex and its immediate vertical
and  %
horizontal neighbours.  %
We know that reachability in planar 
digraphs %
can  %
be reduced to reachability in grid
digraphs %
 in logarithmic space \cite{Allender}. The reduction, however, causes at least a quadratic blow-up in the size of the 
digraph. %
In this paper, we study the simultaneous time--space complexity of reachability in grid graphs.

Asano and Doerr presented a polynomial-time algorithm that uses $O(n^{1/2 + \epsilon})$ space for solving reachability in grid 
digraphs  \cite{Asano11}.  %
Ashida and Nakagawa presented an algorithm with improved space complexity of $\widetilde{O}(n^{1/3})$ \cite{AshidaNakagawa}.
The latter algorithm %
proceeds  %
by first dividing the input grid 
digraph %
into subgrids. It then
uses  %
a gadget to transform each subgrid into a planar 
digraph, %
making the whole of the resultant 
digraph %
planar. Finally, it
uses  %
the planar reachability algorithm of Imai et al. $\cite{Imai}$ as a subroutine to get the desired space bound.

In this paper, we present a $O(n^{1/4 + \epsilon})$ space and polynomial-time algorithm for grid 
digraph %
reachability, thereby significantly improving the
space bound  %
of Ashida and Nakagawa.
\begin{theorem}[Main Theorem]
\label{thm:main}
For every $\epsilon > 0$, there exists a polynomial-time algorithm that
solves  %
reachability in an %
$n$-vertex   %
grid 
digraph %
using $O(n^{1/4 + \epsilon})$ space.
\end{theorem}

The approach of Ashida and Nakagawa \cite{AshidaNakagawa}
is   %
to reduce the size of the input $n$-vertex 
digraph %
to a 
digraph %
of size $O(n^{2/3})$. Their new graph %
preserves  %
reachability between vertices and %
is planar.  %
They use  %
the planar-reachability algorithm of Asano et al. \cite{Asano}.
Our approach is slightly different. We convert the input 
digraph %
into an 
\emph{auxiliary digraph} %
 of size $O(n^{1/2 + \alpha/2})$ for arbitrarily small $\alpha$. The auxiliary 
digraph %
is created by dividing the given grid 
digraph %
into subgrids and replacing paths in each subgrid with a single edge between the boundary vertices.  While our auxiliary 
digraph %
preserves reachability, it is not planar; and hence we cannot use Asano et al.'s algorithm directly. The auxiliary 
digraph %
comes with a drawing with
the  %
\emph{crossing property}, that is, if two edges, $e$ and $f$,
cross each other in this drawing, there necessarily exist two more edges, one from the tail of $e$ to the head of $f$ and the other from the tail of $f$ to the head of $e$, in the 
digraph. %
This property allows us to use a device that we call \emph{pseudoseparator} to solve reachability in it. A pseudoseparator designates a set of vertices, a set of edges and a set of components of the 
digraph, %
such that a path from one component to %
another   %
necessarily either takes a vertex of the pseudoseparator or crosses one of the edges of the pseudoseparator in the drawing.
We finally solve the problem by recursively solving each of these components and using a traversal algorithm. Due to the crossing property, we are required to store only a small number of vertices for performing the traversal, thereby saving space.

In \expref{Section}{sec:prelims}, we state the definitions and
notation  %
that we use in this paper. In \expref{Section}{sec:aux}, we define the auxiliary 
digraph %
and state various properties of it that we use later. In \expref{Section}{sec:pseudosep}, we discuss the concept of the \emph{pseudoseparator}. We give its formal definition and show how to compute the \emph{pseudoseparator} efficiently. In \expref{Section}{sec:auxreach}, we give the algorithm to solve reachability in an auxiliary 
digraph %
and prove its correctness. Finally in \expref{Section}{sec:final} we use the algorithm of \expref{Section}{sec:auxreach} to give an algorithm to decide reachability in grid 
digraphs %
and thus prove \expref{Theorem}{thm:main}.

\section{Preliminaries}
\label{sec:prelims}
Let $[n]$ denote the set $\{0,1,2,\ldots , n\}$. %
We denote the vertex set of a 
digraph %
$G$ by $V(G)$ and its edge set by $E(G)$. We assume that the vertices of a 
digraph %
are \emph{indexed} by integers from $1$ to $\lvert V(G) \rvert$. For a subset $U$ of $V(G)$, we denote the %
sub-digraph %
 of $G$ induced by the vertices of $U$ as $G[U]$ and we denote the %
sub-digraph %
 of $G$ induced by the vertices $V(G)\setminus U$ as $G\setminus U$. For a 
digraph %
$G$, %
we write $\concom{G}$ to denote  %
the set of all connected components
in the underlying undirected graph of $G$.  %
By
the  %
\emph{underlying undirected graph}, we mean the graph formed by %
symmetrizing the adjacency %
relation.  %
To do this, we consider each directed edge in the digraph and create an undirected edge between the corresponding vertices.
Henceforth, whenever we talk about connected components, we will mean the connected components of the underlying undirected graph.
In a \emph{drawing} of a graph 
in the plane
we map each vertex to
a  %
point %
in  %
the plane and each edge to a simple arc whose endpoints coincide with the
images  %
of the end vertices of the edge.
Moreover, 
the image of a vertex does not belong to the interior of an arc corresponding to any edge. %
We say that a graph is \emph{planar} if there exists a drawing of that graph on a plane such that no two
arcs corresponding to the edges %
of the graph %
intersect except at the endpoints. %
Such a drawing is called a \emph{planar embedding}.%

A \emph{planar digraph} is a digraph whose underlying undirected graph is planar. Similarly, a \emph{drawing of a digraph} is defined as the drawing of its underlying undirected graph. %

We will represent a planar graph by describing the
cyclic  %
ordering of a graph's edges around each vertex. %
We note that %
the results in %
\cite{AllenderMahajan} and \cite{Reingold} %
combine to %
a \emph{deterministic logarithmic space} algorithm that decides
whether %
a given graph is planar, and if it is, outputs a planar embedding. Hence, when dealing with planar graphs, we will assume without loss of generality that we have a planar embedding whenever required.   %
An  %
$m\times m$
\emph{grid digraph} %
is a directed graph %
whose 
set of vertices is %
$[m] \times [m] = \{0,1,\ldots,m\}\times\{0,1,\ldots,m\}$ so that if $((i_1,j_1),(i_2,j_2)) $ is an edge then $|i_1-i_2|+|j_1-j_2|=1$.  
In other words, we start with an undirected $m\times m$ grid graph. Then, we
remove some of the edges and assign directions to the remaining edges while keeping all the vertices. This process results in the formation of an $m\times m$ grid digraph, which has a total of $n=(m+1)^2$ vertices. %
 It follows from the definition that grid 
digraphs %
are %
planar.  %

We will work with $m\times m$ grid 
digraphs %
where $m = O(n^{1/2})$. Hence, a
space complexity  %
of $O(m^{1/2 + \epsilon})$ will translate into a space complexity of $O(n^{1/4 + \epsilon'})$ as required.

\section{Auxiliary graph}
\label{sec:aux}
Let $G$ be an $m \times m$ grid 
digraph. %
We divide $G$ into $m^{2\alpha}$ subgrids such that each subgrid is a $m^{1-\alpha} \times m^{1-\alpha}$ grid. Formally, for 
$0 < \alpha < 1$ and %
$1 \leq i,j \leq m^{\alpha}$, the $(i,j)$-th \emph{subgrid} of $G$,
denoted %
$G[i,j]$,  %
is the %
sub-digraph %
of $G$ induced by the set of vertices, $V(G[i,j])$, where $V(G[i, j])$
is %
\[\{(i',j') \mid (i-1)\cdot m^{1-\alpha} \leq i' \leq i\cdot m^{1-\alpha} \textrm{ and } (j-1)\cdot m^{1-\alpha} \leq j' \leq j\cdot m^{1-\alpha} \}\,.\]

For ease of presentation, we will assume without loss of generality that variables like $m$ and $\alpha$ are such that the values like $m^{\alpha}$ and $m^{1-\alpha}$ are integers. %

We define %
the auxiliary 
digraph %
$\aux{G}[i,j]$ as follows. The vertex set of $\aux{G}[i,j]$ is the set of vertices on the \emph{boundary} of $G[i, j]$. 
That means, $V(\aux{G}[i,j])$ is
\begin{align*}
& \{(i',j') \mid (i',j') \in V(G[i,j]) \textrm{ and }\\
& \ \ ((i-1)\cdot m^{1-\alpha} = i' \textrm{ or } i\cdot m^{1-\alpha} = i' \textrm{ or } (j-1)\cdot m^{1-\alpha} = j' \textrm{ or } j\cdot m^{1-\alpha} = j') \}\,.
\end{align*}
For two vertices $u,v$ in $\aux{G}[i,j]$, $(u,v)$ is an edge in $\aux{G}[i,j]$ if there is a path from $u$ to $v$ in the subgrid $G[i,j]$. 
We draw $\aux{G}[i, j]$ on an Euclidean plane by mapping vertex $(x,y)$ to the point with %
coordinates  %
$(x, y)$. The points corresponding to vertices of $\aux{G}[i,j]$ thus lie on a square. We %
use a straight line %
segment  %
to represent the edge if $u$ and $v$ do not lie on a single side
of this square.  %
and an arc 
 inside the
square %
to represent it otherwise.  
We define the %
$\alpha$-\emph{auxiliary digraph}, %
$\aux{G}$ as follows. The vertex set of $\aux{G}$, $V(\aux{G}) = \{(i,j) \mid i = k\cdot m^{1-\alpha} \textrm{ or } j = l\cdot m^{1-\alpha}, \textrm{ such that } 0 \leq k,l \leq m^{\alpha} \}$. The edges of $\aux{G}$ are the edges of $\aux{G}[i,j]$ taken over all pairs $(i,j)$. Note that $\aux{G}$ might have parallel edges, since an edge on a side of a block might be %
in the adjacent block as well. In such cases we preserve both the edges,
in their different blocks of $\aux{G}$ in the drawing of $\aux{G}$ on the plane. \expref{Figure}{fig:auxG} contains an example of a grid 
digraph %
partitioned into subgrids and its corresponding auxiliary 
digraph. %
Since each block $\aux{G}[i,j]$ contains $4m^{1-\alpha}$ vertices, the total number of vertices in $\aux{G}$ %
is  %
at most $4m^{1 + \alpha}$.

Our algorithm for reachability first constructs $\aux{G}$ by solving each of the $m^{1-\alpha} \times m^{1-\alpha}$ grids recursively. It then uses a polynomial-time subroutine to decide reachability in $\aux{G}$. Note that we do not store 
the 
digraph %
$\aux{G}$ explicitly since that would require too much space. Rather we solve a subgrid recursively whenever the subroutine queries for an edge in that subgrid of $\aux{G}$.

Our strategy is to develop a polynomial-time %
algorithm which solves reachability in $\aux{G}$ using $\tilde{O}(\tilde{m}^{1/2 + \beta/2})$ space where $\tilde{m}$ is the number of vertices in $\aux{G}$. As discussed earlier, $\tilde{m}$ %
is %
 at most $4m^{1 + \alpha}$. Hence, the main algorithm requires 
$\tilde{O}(m^{1/2 + \beta/2 + \alpha/2 + \alpha\beta/2})$ space. For a fixed constant $\epsilon > 0$, we can pick $\alpha > 0$ and $\beta > 0$ such that the space complexity becomes $O(m^{1/2 + \epsilon})$.

\begin{figure}[ht]
\centering
\begin{subfigure}{0.4\textwidth}
\begin{tikzpicture}[scale = 0.3]

\foreach \x in {0,...,12}
\foreach \y in {0,...,12}
\filldraw (\x,\y) circle (2pt);
\draw[step=4] (0,0) grid (12,12);

\foreach \x in {0,...,4}
\foreach \y in {0,...,4}
\filldraw (\x,\y) circle (2pt);
\draw[step=4] (0,0) grid (4,4);

\draw[->] (0,1) -- (1, 1);
\draw[->] (1,1) -- (2, 1);
\draw[->] (2,1) -- (3, 1);

\draw[->] (0,2) -- (1, 2);
\draw[->] (1,2) -- (2, 2);
\draw[->] (2,2) -- (3, 2);

\draw[->] (3,1) -- (3, 2);
\draw[->] (3,2) -- (3, 3);

\draw[->] (3,3) -- (4, 3);

\draw[->] (2,1) -- (2, 2);
\draw[->] (2,2) -- (2, 3);
\draw[->] (2,3) -- (2, 4);

\draw[->] (4,3) -- (5,3);
\draw[->] (5,3) -- (6,3);
\draw[->] (6,3) -- (7,3);

\draw[->] (6,3) -- (6,2);
\draw[->] (6,2) -- (6,1);
\draw[->] (6,1) -- (6,0);

\draw[->] (7,3) -- (7,4);

\draw[->] (6,1) -- (7,1);
\draw[->] (7,1) -- (8,1);

\draw[->] (8,1) -- (9,1);
\draw[->] (9,1) -- (9,2);

\draw[->] (9,2) -- (9,3);
\draw[->] (9,3) -- (10,3);
\draw[->] (10,3) -- (10,2);
\draw[->] (10,2) -- (11,2);
\draw[->] (11,2) -- (11,3);
\draw[->] (11,3) -- (11,4);

\draw[->] (2,4) -- (2,5);
\draw[->] (2,5) -- (2,6);
\draw[->] (2,6) -- (2,7);
\draw[->] (2,7) -- (2,8);
\draw[->] (2,6) -- (3,6);
\draw[->] (3,6) -- (4,6);

\draw[->] (4,6) -- (5,6);
\draw[->] (5, 6) -- (6,6);
\draw[->] (6,6) -- (6,7);
\draw[->] (6,7) -- (7,7);
\draw[->] (7,7) -- (8,7);
\draw[->] (6,8) -- (6,7);
\draw[->] (8,5) -- (7,5);
\draw[->] (7,5) -- (8,5); %

\draw[->] (7,5) -- (7,4);
\draw[->] (11,4) -- (11,5);
\draw[->] (11,5) -- (10, 5);
\draw[->] (10,5) -- (9,5);
\draw[->] (9,5) -- (8,5);
\draw[->] (10,5) -- (10,6);
\draw[->] (10,6) -- (10,7);
\draw[->] (10,7) -- (10,8);

\draw[->] (8,7) -- (9,7);
\draw[->] (9,7) -- (10, 7);
\draw[->] (10, 7) -- (11,7);
\draw[->] (11, 7) -- (11,6);
\draw[->] (11,6) -- (12,6);

\draw[->] (0, 10) -- (1, 10);
\draw[->] (1, 10) -- (1, 9);
\draw[->] (1,9) -- (2,9);
\draw[->] (2,9) -- (3,9);
\draw[->] (3,9) -- (3, 10);
\draw[->] (3,10) -- (4,10);
\draw[->] (2,8) -- (2,9);

\draw[->] (4,10) -- (5,10);
\draw[->] (5,10) -- (5,9);
\draw[->] (5, 9) -- (6,9);
\draw[->] (6,9) -- (6,8);
\draw[->] (6,8) -- (6,9); %

\draw[->] (6,10) -- (6,11);
\draw[->] (5,10) -- (6,10);
\draw[->] (6,11) -- (7,11);

\draw[->] (8,9) -- (7,9);
\draw[->] (7,9) -- (7,10);
\draw[->] (7,10) -- (7,11);
\draw[->] (7,11) -- (8,11);

\draw[->] (8,9) -- (9,9);
\draw[->] (9,9) -- (9,10);
\draw[->] (9,10) -- (9,11);
\draw[->] (9,11) -- (8,11);

\draw[->] (10,8) -- (10,9);
\draw[->] (10,9) -- (10,10);
\draw[->] (10,10) -- (10, 11);
\draw[->] (10, 11) -- (11,11);
\draw[->] (11,11) -- (12,11);

\end{tikzpicture}
\end{subfigure}
\begin{subfigure}{0.4\textwidth}
\centering
\begin{tikzpicture}[scale = 0.3]
\foreach \x in {0,...,12}{
\filldraw (\x,0) circle (2pt);
\filldraw (\x,4) circle (2pt);
\filldraw (\x,8) circle (2pt);
\filldraw (\x,12) circle (2pt);
}
\foreach \y in {0,...,12}
{
\filldraw (0,\y) circle (2pt);
\filldraw (4,\y) circle (2pt);
\filldraw (8,\y) circle (2pt);
\filldraw (12,\y) circle (2pt);
}
\draw[step=4] (0,0) grid (12,12);
\draw[->] (0,1) -- (4, 3);
\draw[->] (0,2) -- (2, 4);
\draw[->] (0,1) -- (2, 4);
\draw[->] (0,2) -- (4, 3);

\draw[->] (8,9) to[bend right] (8, 11);
\draw[->] (8,9) to[bend left] (8, 11);

\draw[->] (4,3) -- (7,4);
\draw[->] (4,3) -- (6,0);
\draw[->] (4,3) -- (8,1);
\draw[->] (8,1) -- (11,4);
\draw[->] (2,4) -- (2, 8);
\draw[->] (2,4) -- (4,6);
\draw[->] (4,6) -- (8,7);
\draw[->] (6,8) -- (8,7);
\draw[->] (8,5) -- (7,4);
\draw[->] (11,4) -- (8,5);
\draw[->] (11,4) -- (12,6);
\draw[->] (11,4) -- (10,8);
\draw[->] (8,7) -- (12, 6);
\draw[->] (8,7) -- (10,8);

\draw[->] (2,8) -- (4,10);
\draw[->] (0,10) -- (4,10);

\draw[->] (4,10) -- (6,8);
\draw[->] (4,10) -- (8,11);

\draw[->] (10, 8) -- (12,11);

\end{tikzpicture}

\end{subfigure}
\caption{A grid digraph $G$ divided into subgrids and its corresponding auxiliary digraph $\aux{G}$} %
\label{fig:auxG}

\end{figure}

\subsection{Properties of the auxiliary digraph} %
In the following definition, we give ordered labelling to the vertices of a block of the auxiliary 
digraph. %
We define the labelling with respect to some vertex 
in the block.
\begin{definition}
\label{def:perm}
Let $G$ be a $m \times m$ grid 
digraph, %
$\ell = \aux{G}[i,j]$ be a block of $\aux{G}$ and $v = (x, y)$ be a vertex in $\aux{G}[i,j]$. Let $t = m^{1-\alpha}$. We define a cyclic permutation $c_{\ell}$ on the vertex set of $\aux{G}[i,j]$ as follows:
\[ \ac{l}{}{v} = \begin{cases} (x +1 , y) &\text{ if } x < (i + 1)t \textup{ and } y = jt\\
(x, y +1)&\text{ if } x = (i + 1)t \textup{ and } y < (j + 1)t\\
(x-1, y) &\text{ if } x > it \textup{ and } y = (j + 1)t \\
(x, y -1)&\text{ if } x = it \textup{ and } y > jt \end{cases}\]
For any non-negative integer $r$, we define $c_{\ell}^r$ inductively as follows. For $r = 0$, $\ac{l}{r}{v} = v$ and otherwise we have $\ac{l}{r+1}{v} = \ac{l}{}{\ac{l}{r}{v}}$.
\end{definition}

 Note that the permutation $c_{\ell}$ can be seen as a counter-clockwise shift. Also note that for a block $\ell$ and vertices $v$ and $w$ in it, we can write $v$ as $\ac{l}{p}{w}$ where $p$ is smallest non-negative integer for which $\ac{l}{p}{w} = v$. Next we formalize what it means to say that two edges of the auxiliary %
digraph %
cross each other.

\begin{definition}
\label{def:cross}
Let $G$ be a grid 
digraph %
and $\ell$ be a block of $\aux{G}$. For two distinct edges $e$ and $f$ in the block, such that $e = (v, \ac{l}{p}{v})$ and $f = (\ac{l}{q}{v}, \ac{l}{r}{v})$. We say that edges $e$ and $f$ cross each other if $\min(q, r) < p < \max(q, r)$.
\end{definition}

Note the definition of cross given above is symmetric. That is, if edges $e$ and $f$ cross each other then $f$ and $e$ must cross each other as well. For an edge $f = (\ac{l}{q}{v}, \ac{l}{r}{v})$, we define $\overleftarrow{f} = (\ac{l}{r}{v}, \ac{l}{q}{v})$ and call it the \emph{reverse} of $f$. We also note that if $e$ and $f$ cross each other, then $e$ and $\overleftarrow{f}$ also cross each other.

In \expref{Lemma}{lem:cross} we state an equivalent condition of the crossing of two edges. In \expref{Lemmas}{lem:psg} \expref{and}{lem:closer} we state some specific properties of the auxiliary 
digraph %
that we use later.

\begin{lemma}
\label{lem:cross}
Let $G$ be a grid 
digraph %
and $\ell$ be a block of $\aux{G}$. Let $w$ be an arbitrary vertex in the block $\ell$ and $e = (\ac{l}{p}{w}, \ac{l}{q}{w})$ and $f = (\ac{l}{r}{w}, \ac{l}{s}{w})$ be two distinct edges in $\ell$. Then $e$ and $f$ cross each other if and only if either of the following two conditions hold:
\begin{itemize}
\item $\min(p, q) < \min(r, s) < \max(p, q) < \max(r,s)$
\item $\min(r, s) < \min(p, q) < \max(r, s) < \max(p,q)$
\end{itemize}
\end{lemma}

\begin{proof}
We prove that if $\min(p, q) < \min(r, s) < \max(p, q) < \max(r,s)$ then $e$ and $f$ cross each other. We let $p < r < q < s$. Other cases can be proved by reversing appropriate edges. We thus have integers $r_1 = r - p$, $q_1 = q - p$ and $s_1 = s - p$. Clearly, $r_1 < q_1 < s_1$. Let $v = \ac{l}{p}{w}$. Thus we have $e = (v, \ac{l}{q}{w}) = (v, \ac{l}{q_1}{\ac{l}{p}{w}}) = (v, \ac{l}{q_1}{v})$ and $f = (\ac{l}{r}{w}, \ac{l}{s}{w}) = (\ac{l}{r_1}{\ac{l}{p}{w}}, \ac{l}{s_1}{\ac{l}{p}{w}}) = (\ac{l}{r_1}{v}, \ac{l}{s_1}{v})$
The proof for the second condition is similar.

Now, we prove that if $e = (\ac{l}{p}{w}, \ac{l}{q}{w})$ and $f = (\ac{l}{r}{w}, \ac{l}{s}{w})$ cross each other then either of the given two condition holds. We assume that $p$ is the smallest integer among $p$, $q$, $r$ and $s$. Other cases can be proved similarly. Now, let $v = \ac{l}{p}{w}$. We thus have integers $q_1 = q-p, r_1 = r-p$ and $s_1 = s-p$ such that $e = (v, \ac{l}{q_1}{v})$ and $f = (\ac{l}{r_1}{v}, \ac{l}{s_1}{v})$. Since $e$ and $f$ cross each other, we have $\min(r_1, s_1) < q_1 < \max(r_1, s_1)$. Thus $\min(r_1+p, s_1+p) < q_1+p < \max(r_1+p, s_1+p)$. It follows that $\min(r, s) < q < \max(r,s)$. Since we assumed $p$ to be smallest integer among $p$, $q$, $r$ and $s$; we have $\min(p, q) < \min(r, s) < \max(p, q) < \max(r,s)$, thus proving the lemma.
\end{proof}

We see that we can draw an auxiliary %
digraph %
on a plane such that the arcs corresponding to two of its edges intersect if and only if the corresponding edges cross each other. Henceforth, we will work with such a drawing.
\begin{definition}
Let $G$ be a grid %
digraph %
 and $\ell$ be a block of $\aux{G}$. For a vertex $v$ and edges $f$, $g$ such that $f = (\ac{l}{q}{v}, \ac{l}{r}{v})$ and $g = (\ac{l}{s}{v}, \ac{l}{t}{v})$, we say that $f$ is closer to $v$ than $g$ if $\min(q, r) < \min(s, t)$.

We say $f$ is \emph{closest} to $v$ if there exists no other edge $f'$ which is closer to $v$ than $f$.
\end{definition}

\begin{lemma}
\label{lem:psg}
Let $G$ be a grid %
digraph %
 and $e_1 = (u_1,v_1)$ and $e_2 = (u_2,v_2)$ be two edges in $\aux{G}$. If $e_1$ and $e_2$ cross each other, then $\aux{G}$ also contains the edges $(u_1, v_2)$ and $(u_2, v_1)$.
\end{lemma}
\begin{proof}
Let $e_1 = (v, \ac{l}{p}{v})$ and $e_2 = (\ac{l}{q}{v}, \ac{l}{r}{v})$ be two edges that cross each other in $\aux{G}$. Let $\ell$ be the block of $\aux{G}$ to which $e_1$ and $e_2$ belong. Consider the subgrid of $G$ which is solved to construct the block ${\ell}$. Since the edge $e_1$ exists in block ${\ell}$, there exists a path $P$ from $v$ to $\ac{l}{p}{v}$ in the underlying subgrid. This path $P$ divides the subgrid into two parts such that the vertices $\ac{l}{q}{v}$ and $\ac{l}{r}{v}$ belong to different parts of the subgrid. Thus, a path between $\ac{l}{q}{v}$ and $\ac{l}{r}{v}$ necessarly take a vertex of path $P$. Hence, there is a path from $v$ to $\ac{l}{r}{v}$ and a path from $\ac{l}{p}{v}$ to $\ac{l}{r}{v}$. Thus the lemma follows.
\end{proof}

\begin{figure}[ht]
\begin{center}
\includegraphics[width=0.3\textwidth]{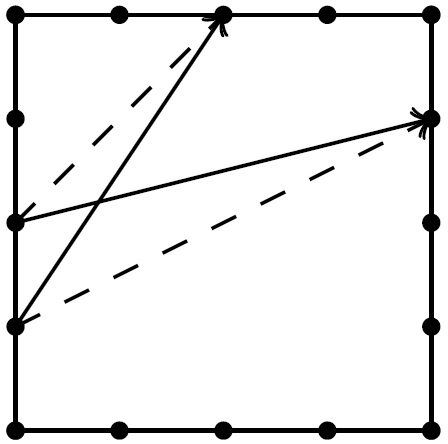}
\hfill
\includegraphics[width=0.3\textwidth]{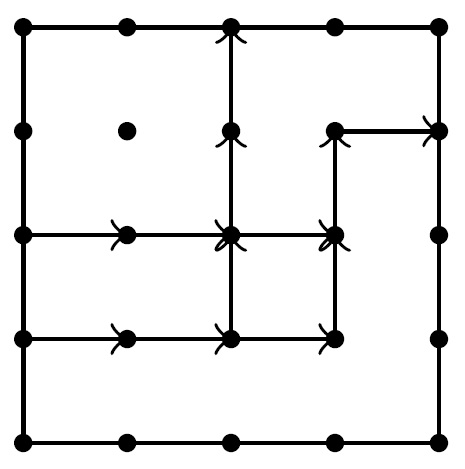}

\end{center}
\caption{Edge crossings in an auxiliary grid. On the left, there is one block of the auxiliary %
digraph %
 that contains edges that cross. The dotted edges are the ones whose existence is made necessary by \expref{Lemma}{lem:psg}. On the right, a subgrid which results in the auxiliary block on the left.}
\end{figure}

\begin{lemma}
\label{lem:closer}
Let $G$ be a grid %
digraph %
 and $e_1 = (u_1,v_1)$ and $e_2 = (u_2,v_2)$ be two edges in $\aux{G}$. If $e_1$ and $e_2$ cross a certain edge $f = (x,y)$, and $e_1$ is $closer$ to $x$ than $e_2$, then the edge $(u_1, v_2)$ is also %
 in $\aux{G}$.
\end{lemma}

\begin{figure}
\centering
\includegraphics[width=0.9\textwidth]{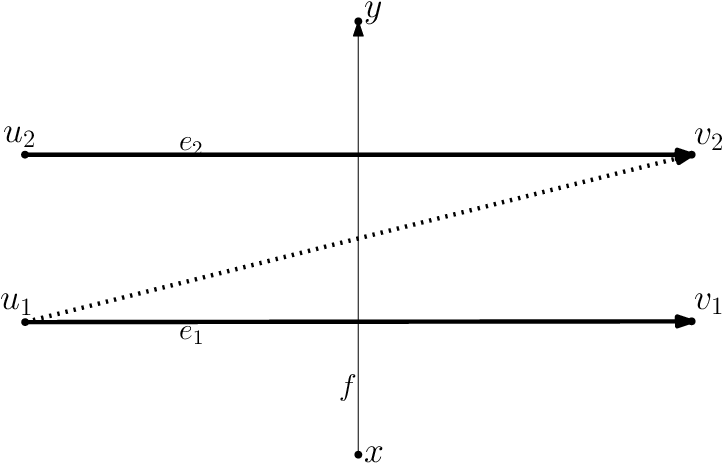}

\caption{Illustration of \expref{Lemma}{lem:closer}}
\end{figure}

\begin{proof}
Let $f = (v, c_{\ell}^p(v))$, $e_1 = (c_{\ell}^q(v), c_{\ell}^r(v))$ and $e_2 = (c_{\ell}^s(v), c_{\ell}^t(v))$. If $c_{\ell}^q(v) = c_{\ell}^s(v)$ then the lemma trivially follows. Otherwise, we have two cases to consider:
\begin{description}
\item[Case 1 ($e_1$ crosses $e_2$):] In this case, we will have $(c_{\ell}^q(v), c_{\ell}^t(v))$ %
in $\aux{G}$ by \expref{Lemma}{lem:psg}.
\item[Case 2 ($e_1$ does not cross $e_2$):] In this case, we have $\min(q,r) < \min(s,t) < p < \max(s, t) < \max(q, r)$. Since $e_1$ crosses $f$, we have the edge $(c_{\ell}^q(v), c_{\ell}^p(v))$ in $\aux{G}$ by 
\expref{Lemma}{lem:psg}. This edge will cross $e_2$. Hence $(c_{\ell}^q(v), c_{\ell}^t(v))$ is %
 in $\aux{G}$.\qedhere
\end{description}
\end{proof}

\section{Pseudoseparators in a grid digraph}  %
\label{sec:pseudosep}
Imai et al. used a separator construction to solve the reachability problem in planar %
digraphs %
 \cite{Imai}. A separator is a set of vertices whose removal disconnects the graph into \emph{small components}. The class of grid %
digraphs %
is a subclass of planar %
digraphs. %
Grid graphs are known to have small separators. However, for a grid %
digraph %
 $G$, the %
underlying undirected graph of %
 $\aux{G}$ might not have a small separator.

An essential property of a separator is that, for any two vertices, a path between the vertices must contain a separator vertex if the vertices lie in two different components with respect to the separator. This property can then be used to design a divide and conquer %
algorithm for reachability where only separator vertices need to be marked and stored for traversal. \expref{Lemma}{lem:closer} allows us to use edges as well since, at most, one vertex on each side of an edge needs to be stored for traversal (the visited-vertex closest to the tail of the edge). Hence, we can use a slightly weaker structure in place of a separator. We construct a structure called pseudoseparator (see \expref{Definition}{def:psc}). It is essentially a set of vertices and edges that can be used to divide the %
digraph %
into components, such that a path from one component to another either takes a vertex of the pseudoseparator or crosses an edge of the pseudoseparator.

For a %
digraph %
$H = (V_1, E_1)$ given along with its drawing,
and a  %
sub-digraph %
$C = (V_2, E_2)$ of $H$, define the 
digraph %
$\sep{H}{C} = (V_3, E_3)$ as $V_3 = V_1 \setminus V_2$ and $E_3 = E_1 \setminus \{e \in E_1 \mid \exists f \in E_2, e \textrm{ crosses } f\}$. We note that the %
digraph %
$H$ we will be working with throughout the article will be a  %
sub-digraph %
 of an auxiliary %
digraph. %
Hence it will always come with a drawing.

\begin{definition}
\label{def:psc}
Let $G$ be a grid %
digraph %
and $H$ be
an induced  %
sub-digraph %
 of $\aux{G}$ with $h$ vertices. Let $f:\mathbb{N} \rightarrow \mathbb{N}$ be a function. A  %
sub-digraph %
 $C$ of $H$ is said to be an $f(h)$-$\psc$ of $\aux{G}$ if the size of every connected component in $\concom{\sep{H}{C}}$ is at most $f(h)$. The size of $C$ is the
number of vertices plus the number of edges of $C$.
\end{definition}

For
an induced%
sub-digraph %
 $H$ of $\aux{G}$, an $f(h)$-$\psc$ is a  %
sub-digraph %
 $C$ of $H$ that has the property that, if we remove the vertices as well as all the edges that cross one of the edges of the $\psc$, the %
digraph %
 gets disconnected into small pieces. Moreover for every edge $e$ in $H$, if there exists distinct sets $U_1$ and $U_2$ in $\concom{\sep{H}{C}}$ such that one of the endpoints of $e$ is in $U_1$ and the other is in $U_2$, then there exists an edge $f$ in $C$ such that $e$ crosses $f$. Hence any path in $H$ that connects two vertices in different components of $H \diamond C$ must either contain a vertex of $C$ or must contain an edge that crosses an edge of $C$. We divide the %
digraph %
 using this $\psc$ and give an algorithm that recursively solves each  %
sub-digraph %
 and then combines their solution efficiently.
\subsection{Constructing a pseudoseparator}
Before working out the details below, we briefly comment on how to construct a $\psc$ of
an induced%
sub-digraph %
 $H$ of $\aux{G}$. Ideally, we would have liked to have a set of cycles of $H$
such that the surfaces obtained by cutting along these cycles have %
a bounded number of vertices 
on it.%
However, such cycles might not exist. We instead pick a maximal subset of edges from $H$ so that no two edges cross (see \expref{Definition}{def:planar} and \expref{Lemma}{lem:planarG}). Then we triangulate the resulting graph. We show that this triangulated graph has the required cycles and can be found using Imai et al.'s separator algorithm. Call the triangulated graph %
$\trplanar{H}$ and the separator vertices as $S$. Since the cycles might have triangulated edges, we will add at most a constant number of vertices and edges from $H$ to \emph{shield} these triangulated edges. The vertex set of $\psc$ of $H$ will thus contain all the vertices of $S$ and at most four additional vertices for each edge of $\trplanar{H}[S]$ that is not %
 in $H$. The edge set of $\psc$ of $H$ will contain all the edges of $H$ which are also in $\trplanar{H}[S]$ and at most four additional edges for each edge of $\trplanar{H}[S]$ that is not %
 in $H$.

\begin{definition}
\label{def:planar}
Let $G$ be a grid %
digraph %
 and $H$ be
an induced  %
sub-digraph %
 of $\aux{G}$. We define $\mxplanar{H}$ as a  %
sub-digraph %
 of $H$. The vertex set of $\mxplanar{H}$ is same as that of $H$. For an edge $e \in H$, let ${\ell}$ be the block to which $e$ belongs and let $w$ be the lowest indexed vertex in that block. Then $e = (c^i_{\ell}(w), c^j_{\ell}(w))$ is in $\mxplanar{H}$ if there exists no other edge $f = (c^x_{\ell}(w), c^y_{\ell}(w))$ in $H$ such that $\min(x,y) < \min(i, j) < \max(x, y) < \max(i, j)$.
\end{definition}

In \expref{Lemma}{lem:planarG} we show that the %
digraph %
 $\mxplanar{H}$ is indeed planar. We also prove that the subset of edges that we have picked from $H$ is a maximal subset such that no two edges cross.

\begin{lemma}
\label{lem:planarG}
Let $G$ be a grid %
digraph %
and $H$ be
 an induced  %
sub-digraph %
 of $\aux{G}$. No two edges of $\mxplanar{H}$ cross each other. Moreover, for any edge $e$ in $H$ that is not in $\mxplanar{H}$, there exists another edge in $\mxplanar{H}$ that crosses $e$.
\end{lemma}
\begin{proof}
Let ${\ell}$ be a block of $\aux{G}$ and $w$ be the smallest index vertex of ${\ell}$. Let $e = (c^p_{\ell}(w), c^q_{\ell}(w))$ and $f = (c^r_{\ell}(w), c^s_{\ell}(w))$ be two edges of $H$ that cross. We have, by \expref{Lemma}{lem:cross}, that either $\min(p,q) < \min(r, s) < \max(p, q) < \max(r, s)$ or $\min(r,s) < \min(p, q) < \max(r, s) < \max(p, q)$. Hence, by our construction of $\mxplanar{H}$, atmost one of $e$ and $f$ belongs to it. Thus no two edges of $\mxplanar{H}$ %
 cross.

For the second part, we will prove by contradiction. Let us assume that there exists edges in $H$ which is not in $\mxplanar{H}$ and also not crossed by an edge in $\mxplanar{H}$. We pick edge $e = (c_{\ell}^p(w), c_{\ell}^q(w))$ from them such that $\min(p, q)$ of that edge is minimum. Since this edge is not %
 in $\mxplanar{H}$, we have by definition, an edge $f = (c_{\ell}^r(w), c_{\ell}^s(w))$ such that $\min(r,s) < \min(p, q) < \max(r, s) < \max(p, q)$. We pick the edge $f$ for which $\min(r,s)$ is minimum. Now, since this edge $f$ is not %
 in $\mxplanar{H}$, we have another edge $g = (c_{\ell}^i(w), c_{\ell}^j(w))$ in $\mxplanar{H}$ such that $\min(i,j) < \min(r, s) < \max(i, j) < \max(r, s)$. We pick $g$ such that $\min(i, j)$ is minimum and break ties by picking one whose $\max(i,j)$ is maximum. Now, we have the following cases:
\begin{description}
\item[Case 1 ($i < r < j < s$):] Note that $H$ is
an induced  %
sub-digraph %
 of $\aux{G}$. Thus, in this case, the edge $(c_{\ell}^i(w), c_{\ell}^s(w))$ will be %
 in $H$ due to \expref{Lemma}{lem:psg}. Since $i < p < s < q$, and $i< \min(r, s)$, this will contradict the way in which edge $f$ was chosen.
\item[Case 2 ($i < s < j < r$):] In this case, the edge $(c_{\ell}^r(w), c_{\ell}^j(w))$ will be %
 in $H$ due to \expref{Lemma}{lem:psg}. This edge will cross $e$ and hence not be %
 in $\mxplanar{H}$. Thus, we have an edge $g' = (c_{\ell}^{i'}(w), c_{\ell}^{j'}(w))$ in $\mxplanar{H}$ such that $\min(i', j') < j < \max(i',j') < r$. We will thus have two subcases.
\begin{description}
\item[Case 2a ($i < \min(i', j')$):] Here, we will have $i < \min(i', j') < j < \max(i', j')$. Hence this edge will cross $g$ giving a contradiction to the first part of this Lemma.
\item[Case 2b ($\min(i', j') \leq i$):] Here, this edge should have been chosen instead of $g$ contradicting our choice of $g$.
\end{description}
\end{description}
The analysis of two remaining cases where $j<s<i<r$ and $j < r < i < s$ are similar to Cases $1$ and $2$ respectively.
\end{proof}

\begin{definition}
Let $G$ be a grid %
digraph %
and $H$ be
an induced  %
sub-digraph %
 of $\aux{G}$. The graph $\trplanar{H}$ is formed by first running \expref{Algorithm}{alg:completeboundary} on $H$ and then triangulating the %
underlying undirected graph of the result%
\end{definition}

\begin{algorithm}
\LinesNumbered
\caption{Completing the boundary of $\mxplanar{H}$}
\label{alg:completeboundary}
\KwIn{ 
An %
induced  %
sub-digraph %
 $H$ of $\aux{G}$}

Output all the vertices of $H$\;
Output all the edges of $\mxplanar{H}$\;
\ForEach{block $\ell$ in $H$}{
	\ForEach{vertex $v$ of $H$ in $\ell$}{
		 $p \gets$ the smallest positive integer such that $c_{\ell}^p(v) \in V(H)$\;
		\If{$(v, c_{\ell}^p(v)) \notin E(\mxplanar{H})$}{
			Output the edge $(v, c_{\ell}^p(v))$\;
		} 
}
}
\end{algorithm}

Note that \expref{Algorithm}{alg:completeboundary} completes the boundary of each block. Thus, for any block, the vertices in its boundary form a simple cycle. Since the interiors of these cycles are disjoint, there is no edge that is drawn through more than one block in $\trplanar{H}$. Thus, each edge of $\trplanar{H}$ %
lies  %
either entirely inside \emph{one} of the blocks, or
lies  %
completely outside the whole $m \times m$ grid.

Now, for each of those edges added to $\trplanar{H}$ as part of triangulation that is inside one of the blocks, we define a set of at most four \emph{shield} edges.

\begin{definition}Let $G$ be a grid %
digraph %
 and $H$ be %
an induced  %
sub-digraph %
 of $\aux{G}$. Let $e = (v, w)$ be an edge in $\trplanar{H}$ such that $e$ is inside one of the blocks and $e$ is not %
 in $\mxplanar{H}$. Let $p$ and $q$ be integers such that $w = c_{\ell}^{p}(v)$ and $v = c_{\ell}^q(w)$.

\begin{itemize}
\item Let $p_1$ be the largest integer such that $p_1 < p$ and an edge $e_1$ with endpoints $v$ and $c_{\ell}^{p_1}(v)$ exists in $\mxplanar{H}$. $e_1$ is undefined if no such integer exists.
\item Let $p_2$ be the smallest integer such that $p_2 > p$ and an edge $e_2$ with endpoints $v$ and $c_{\ell}^{p_2}(v)$ exists in $\mxplanar{H}$. $e_2$ is undefined if no such integer exists.
\item Let $q_1$ be the largest integer such that $q_1 < q$ and an edge $e_3$ with endpoints $c_{\ell}^{q_1}(w)$ and $w$ exists in $\mxplanar{H}$. $e_3$ is undefined if no such integer exists.
\item Let $q_2$ be the smallest integer such that $q_2 > q$ and an edge $e_4$ with endpoints $c_{\ell}^{q_2}(w)$ and $w$ exists in $\mxplanar{H}$. $e_4$ is undefined if no such integer exists.
\end{itemize}
For $i = 1,2,3,4$, the edges $e_i$ which are defined above are called \emph{shield} edges of $e$.
\end{definition}

We will be using the following Lemma, which was proven by Imai et al., to help us in the construction of $\psc$.

\begin{lemma}\label{lem:seplemma}\cite{Imai}
For every $\beta > 0$, there exists a
polynomial-time   %
and $\tilde{O}(h^{1/2 + \beta/2})$-space %
algorithm that takes a $h$-vertex planar graph $P$ as input and outputs a set of vertices $S$, such that $|S|$ is $O(h^{1/2 + \beta/2})$ and removal of $S$ disconnects the graph into components of size $O(h^{1-\beta})$.
\end{lemma}

\expref{Algorithm}{alg:psepH} describes how to construct a  %
sub-digraph %
 of $H$ which we call $\psep(H)$. In \expref{Lemma}{lem:psep} we show that $\psep(H)$ is a $\psc$ of $H$.

\begin{algorithm}
\LinesNumbered
\caption{Construction of \psep(H)}
\label{alg:psepH}
\KwIn{%
An%
induced  %
sub-digraph %
 $H$ of $\aux{G}$ and a positive number $\beta$}
\KwOut{The %
digraph %
 \psep(H)}

Construct $\trplanar{H}$ from $H$\;
Find a set $S$ of vertices in $\trplanar{H}$ which divides it into components of size $O(n^{1-\beta})$ by applying \expref{Lemma}{lem:seplemma} \label{line:spacedom}\;
Output all the vertices of $S$\; 
Output those edges of %
$H$ which are also (in its undirected form) in $\trplanar{H}[S]$\;

\ForEach{edge $e = (v, w)$ of $\trplanar{H}$}{
\If{$e$ is \emph{inside} a block of $\trplanar{H}$ \textbf{and} $e \notin E(\mxplanar{H})$ \textbf{and} both endpoints of $e$ are in $S$}{
Output all the shield edges of $e$ along with their end vertices.
}
}
\end{algorithm}

\begin{lemma}\label{lem:psepcomplexity}
Let $G$ be a grid %
digraph %
, and $H$ be %
an%
 induced  %
sub-digraph %
 of $\aux{G}$. For every $\beta > 0$, there exists a polynomial-time %
and $\tilde{O}(h^{1/2 + \beta/2})$-space  %
algorithm that takes $H$ as input and outputs \psep(H).
\end{lemma}
\begin{proof}
  To prove this lemma, we analyse the space and time complexity of 
\expref{Algorithm}{alg:psepH}. We will first see that $\trplanar{H}$ %
is %
 constructed from $H$ in logspace. To see it, we first note that we can decide if an edge $e$ of $H$ belong to $\mxplanar{H}$ in logspace by checking if $e$ satisfies the condition required by \expref{Definition}{def:planar}. Next, the boundary of $\mxplanar{H}$ %
is %
 completed in logspace by using \expref{Algorithm}{alg:completeboundary}. Finally we triangulate the resultant graph in logspace: For every face $f$ in the resultant graph, we connect each vertex of $f$ to the lowest-indexed vertex of $f$. The space complexity of construction of \psep(H) is thus dominated by the space required by \expref{step}{line:spacedom} of \expref{Algorithm}{alg:psepH}. By \expref{Lemma}{lem:seplemma}, we know that this step requires $O(h^{1/2 + \beta/2})$ space. Hence, the space required by \expref{Algorithm}{alg:psepH} is $O(h^{1/2 + \beta/2})$. Each step of \expref{Algorithm}{alg:psepH} %
is %
 performed in polynomial time, hence the total time complexity is bounded by a polynomial.
\end{proof}

\begin{lemma}
\label{lem:psep}
Let $G$ be a grid %
digraph %
, and $H$ be %
an%
induced  %
sub-digraph %
 of $\aux{G}$. The %
digraph %
 $\psep(H)$ is a $h^{1-\beta}$-$\psc$ of $H$.
\end{lemma}
To prove \expref{Lemma}{lem:psep}, we first show a property of triangulated graphs that we use in our construction of $\psc$. We know that a simple cycle in a planar embedding of a planar graph divides the plane into two parts. We call these two parts the two \emph{sides} of the cycle.

\begin{lemma}
\label{lem:sepfamily}
Let $G$ be a triangulated planar graph, and $S$ be a subset of its vertices. For each pair of vertices $u, v$ belonging to different components of $G \setminus S$, a cycle in $G[S]$ exists, such that $u$ and $v$ belong to different sides of this cycle.
\end{lemma}

\begin{proof}
To prove the Lemma, we first need some terminology. We call a set of faces a \emph{region of edge-connected faces} if they induce a connected subgraph in the dual graph. We can orient the edges of an undirected planar simple cycle to make it a directed cycle. This can help us identify the two \emph{sides} of the cycle as \emph{interior} (left-side) and \emph{exterior} (right-side).

Let $C$ be a component of $G \setminus S$ and $S'$ be the set of vertices of $S$ adjacent to at least one of the vertices of $C$ in $G$. Let $F$ be the set of triangle faces of $G$ to which at least one vertex of $C$ belongs. Let $\tilde{G}$ be the dual graph of $G$ and $\tilde{G}[F]$ be the subgraph induced by $F$ on this dual graph.  %
We first observe that since we have started with a triangulated graph, the vertices of any face $f$ of $F$ will either belong to $C$ or $S'$. 

Let $f_1$ and $f_2$ be two arbitrary faces of $F$. We will first show that $F$ is a region of edge-connected faces by showing that there exists a path between $f_1$ and $f_2$ in $\tilde{G}[F]$. Let $v_1$ be a vertex of $C$ which belongs to $f_1$. Similarly, let $v_2$ be a vertex of $C$ which belongs to $f_2$. Let the length of shortest path between $v_1$ and $v_2$ in $G \setminus S$ be $k$. We prove our claim by induction on $k$. If $k = 0$, then $v_1 = v_2$. We know that the set of faces that share a vertex induce a connected subgraph in the dual graph, hence our claim holds for $k=0$.
Now, we assume that our claim holds for path length $k = \ell-1$. To prove that our claim holds for $k= \ell$, let $(v_3, v_2)$ be the last edge in the shortest path from $v_1$ to $v_2$. Let $f_3$ be a face with both $v_3$ and $v_2$ in its boundary. Since the length of shortest path from $v_1$ to $v_3$ is $\ell-1$, by induction hypothesis, there exists a path between $f_1$ and $f_3$ in $\tilde{G}[F]$. Since $f_3$ and $f_2$ share the vertex $v_2$, there is a path between $f_3$ and $f_2$ in $\tilde{G}[F]$. Combining these two paths, we get a path between $f_1$ and $f_2$ in $\tilde{G}[F]$ which proves our claim.

Miller proved that the boundary of a region of edge-connected faces is a set of edge-disjoint \emph{simple} cycles which can be oriented such that they have disjoint exteriors \cite{Miller}. We claim that all the vertices of any boundary cycle of $F$ are contained in $S'$. We prove this claim by contradiction. Let us assume that $v$ is a vertex in a boundary cycle of $F$ that is not %
 in $S'$. In this case, $v$ belongs to $C$. Consider an edge $(v,w)$ of the boundary cycle. Let $f_a$ and $f_b$ be the two faces that shares the edge $(v,w)$. Let $v_a$ and $v_b$ be the third vertex of $f_a$ and $f_b$ respectively. Since $v_a$, $v_b$ and $w$ are all connected to $v$, they are %
 in either $S'$ or $C$. Thus both $f_a$ and $f_b$ are in $F$. However, this contradicts that $(v,w)$ is at boundary of $F$. 

Thus, we have proved that the vertices of a connected component of $G\setminus S$ are contained inside a set of edge-disjoint simple cycles with disjoint exteriors. The vertices of all such cycles are contained in $S$. Hence the lemma follows.
\end{proof}

\begin{proof}[Proof of \expref{Lemma}{lem:psep}]
Let $C = \psep(H)$. Let $S$ be the set of vertices obtained from $\trplanar{H}$ by using \expref{Lemma}{lem:seplemma}. We claim that if any two vertices $u$ and $v$ belong to different connected components of $\trplanar{H} \setminus S$, then it belongs to different components of $\concom{\sep{H}{C}}$. We prove this by contradiction. Let us assume that it is not true. Then there is an edge $e$ in $H$ and two distinct sets $U_1$ and $U_2$ of $\concom{\trplanar{H}\setminus S}$ such that one of the end point of $e$ is in $U_1$, the other is in $U_2$ and $e$ does not cross any of the edges of $\psep(H)$. Without loss of generality, let $e = (v, c_{\ell}^p(v))$, for some block $\ell$, where $v \in U_1$ and $c_{\ell}^p(v)$ is not in $U_1$ (we pick the edge $e$ such that $p$ is minimum for any choice of $v$). Due to \expref{Lemma}{lem:sepfamily}, it follows that there exists an edge $f$ in $\trplanar{H}[S]$ such that $f = ((c_{\ell}^q(v)), c_{\ell}^r(v))$ and that $e$ crosses $f$. This edge $f$ is a triangulation edge, for otherwise it is also %
 in \psep(H) giving us a contradiction. We orient the triangulation edge so that $q < p < r$. Now, since $e$ is not %
 in $\mxplanar{H}$, by \expref{Lemma}{lem:planarG} there exists at least one edge that crosses it and is %
 in $\mxplanar{H}$. Let $g = (c_{\ell}^s(v), c_{\ell}^t(v))$ be one such edge such that $t - s$ is maximum\footnote{We note that we do \emph{not} consider the \emph{absolute} value of the $t-s$, rather we consider the actual value. $t-s$ can hence be negative for an edge.} We thus have the following cases:
\begin{description}
\item[Case 1 ($s<q<p<r<t$):] In this case, since $g$ crosses $e$, by
\expref{Lemma}{lem:psg}, we have that the edge $e' = (c_{\ell}^s(v), c_{\ell}^p(v))$ is also %
 in $H$. $e'$ also crosses $f$. Since $p-s < p$, existence of $e'$ contradicts our choice of $e$.
\item[Case 2 ($q<t<p<s<r$):] In this case, since $g$ crosses $e$, by 
\expref{Lemma}{lem:psg}, we have that the edge $e' = (v, c_{\ell}^t(v))$ is also %
 in $H$. $e'$ also crosses $f$. Since $t < p$, existence of $e'$ contradicts our choice of $e$.
\item[Case 3 ($t<q<p<r<s$):] In this case, the edge $e' = (c_{\ell}^{s}(v), c_{\ell}^p(v))$ will also be %
 in $H$. $e'$ will cross $f$ and hence will not be %
 in $\mxplanar{H}$. By \expref{Lemma}{lem:planarG}, there exists an edge $g' = (c_{\ell}^{s'}(v), c_{\ell}^{t'}(v))$ in $\mxplanar{H}$ that crosses $e'$. Since both $g'$ and $g$ are %
 in $\mxplanar{H}$, these edges %
do not %
cross each other. 

Using the fact that $g'$ crosses $e'$ and $g'$ does not cross $g$, we get the following:

\[t \leq \min(s',t') < p < \max(s', t') < s\]

Consequently, $t'-s' > t-s$ and $g'$ crosses $e$. This contradicts our choice of $g$.

\item[Case 4 ($q<s<p<t<r$):] In this case, since $g$ crosses $e$, by
\expref{Lemma}{lem:psg}, we have that the edge $e' = (v, c_{\ell}^t(v))$ is %
 in $H$. $e'$ also crosses $f$ and hence $e'$ was not %
 in $\mxplanar{H}$. Thus there will exist an edge in $\mxplanar{H}$ which crosses $e'$ by \expref{Lemma}{lem:planarG}. Let $g' = (c_{\ell}^{s'}(v), c_{\ell}^{t'}(v))$ be the edge in $\mxplanar{H}$ that crosses $e'$ such that $t'-s'$ is maximum. We have the following subcases:
\begin{description}
\item[Case 4a ($s'<q<t<r<t'$):]In this case, since $g'$ crosses $e$, by \expref{Lemma}{lem:psg}, we have that the edge $e'' = (c_{\ell}^{s'}(v), c_{\ell}^p(v))$ is also %
 in $H$. $e''$ also crosses $f$. Since $p-s' < p$, existence of $e''$ contradicts our choice of $e$.
\item[Case 4b ($q<t'<t<s'<r$):] In this case, we see that since $g$ and $g'$ are both %
 in $\mxplanar{H}$, they
do not %
cross each other. Therefore. $t' \leq s$ and consequently $t' < p$. Thus, $g'$ crosses $e$ and by \expref{Lemma}{lem:psg}, the edge $e'' = (v, c_{\ell}^{t'}(v))$ is also %
 in $H$. $e''$ also crosses $f$. Since $t' < p$, existence of $e''$ contradicts our choice of $e$.
\item[Case 4c ($t'<q<t<r<s'$):] In this case, the edge $e'' = (c_{\ell}^{s'}(v), c_{\ell}^t(v))$ will also be %
 in $H$. $e''$ will cross $f$ and hence will not be %
 in $\mxplanar{H}$. By \expref{Lemma}{lem:planarG}, there exists an edge $g'' = (c_{\ell}^{s''}(v), c_{\ell}^{t''}(v))$ in $\mxplanar{H}$ that crosses $e''$. Since $g''$, $g'$ and $g$ are all %
 in $\mxplanar{H}$, these edges 
do not %
cross each other. 

Using the fact that $g''$ crosses $e''$, $g''$ does not cross $g'$ and $g''$ does not cross $g$; we get the following:

\[t' \leq \min(s'',t'') <s< t < \max(s'', t'') < s'\]

Consequently, $t''-s'' > t'-s'$ and $g''$ crosses $e'$. This contradicts our choice of $g'$.

\item[Case 4d ($q<s'<t<t'<r$):] In this case, we see that since $g$ and $g'$ are both %
 in $\mxplanar{H}$, they
do not %
cross each other. Therefore $s' \leq s$ and consequently $s' < p$. Thus, $g'$ crosses $e$ and $t' - s' > t-s$. This contradicts our choice of $g$.

\end{description}
\end{description}
In other cases, if $g$ is picked such that one of its vertices is common with $f$, then $e$ will cross a $shield$ edge of $f$, giving a contradiction. If $g$ is picked such that $g$ cross $f$, then it contradicts the fact that both of them are %
 in $\trplanar{H}$, which is a planar %
digraph %
.
\end{proof}
Summarizing \expref{Lemmas}{lem:psepcomplexity} \expref{and}{lem:psep} we have \expref{Theorem}{thm:constructpsc}.
\begin{theorem} \label{thm:constructpsc}
Let $G$ be a grid %
digraph %
 and $H$ be %
an%
induced subgraph of $\aux{G}$ with $h$ vertices. For any constant $\beta > 0$, there exists an $\tilde{O}(h^{1/2 + \beta/2})$-space and polynomial-time %
algorithm that takes $H$ as input and outputs an $h^{1-\beta}$-$\psc$ of size $O(h^{1/2 + \beta/2})$.
\end{theorem}

\section{Algorithm to solve reachability in the auxiliary %
digraph %
}   %
\label{sec:auxreach}
	\begin{figure}
		\centering
		\begin{subfigure}{0.4\textwidth}
			\begin{tikzpicture}[scale = 0.5]
			
			\filldraw (-1,2) circle (2pt) node[below] {$u$};    
			\filldraw (1,1) circle (2pt);    
			\filldraw (1,2) circle (2pt);    
			\filldraw (1,3) circle (2pt) node[above] {$w$};    
			\filldraw (1,4) circle (2pt);
			\filldraw (3,1) circle (2pt) node[below] {$v$};	
			
			\draw (-2,2) circle [x radius=1.5cm, y radius=10mm];		
			\draw[->] (-1,2) -- (1,3);
			
			\draw (4,1) circle [x radius=1.5cm, y radius=10mm];	
			\draw[->] (1,3) -- (3,1);
			
			\draw (1,2.5) circle [x radius=0.5cm, y radius=20mm];

			\end{tikzpicture}
			
			\caption{\scriptsize The $s$-$t$ path takes a vertex of the separator}
			\label{fig:viaseparator}
			
		\end{subfigure}
	\hfill
		\begin{subfigure}{0.4\textwidth}
			\centering
			\begin{tikzpicture}[scale = 0.5]
			
			\filldraw (-1,2) circle (2pt) node[left] {$u'$};    
			\filldraw (1,1) circle (2pt) ;
			\filldraw (1,4) circle (2pt) ;
			\filldraw (3,2) circle (2pt) node[right] {$v'$};			
			
			\filldraw (-1,3) circle (2pt) node[left] {$u$};			
			\filldraw (3,3) circle (2pt) node[right] {$v$};			
			
			\draw[->] (-1,2) -- (3,2);
			
			\draw (-2,2.5) circle [x radius=1.5cm, y radius=10mm];
			\draw (1,2.5) circle [x radius=0.5cm, y radius=20mm];	
			\draw (4,2.5) circle [x radius=1.5cm, y radius=10mm];

			\draw[->] (-1,3) -- (3,3);
			\draw[->] (1,1) --node[right] {$e$} (1,4) ;
			\draw[dashed, ->] (-1, 2) -- (3,3);

			\end{tikzpicture}
			
			\caption{\scriptsize The $s$-$t$ path crosses an edge of the separator}
			\label{fig:viaedgecrossing}
		\end{subfigure}
		
		\caption{ }
		
	\end{figure}

In this section, we discuss the grid %
digraph %
 reachability algorithm. Let $G$ be a grid %
digraph %
 having $\widetilde{n}$ vertices. By induction, we assume that we have access to %
an %
induced subgraph $H$ of $\aux{G}$, containing $h$ vertices. Below we describe a recursive procedure $\auxreach{H}{x}{y}$ that outputs \textsf{true} if there is a path from $x$ to $y$ in $H$ and outputs \textsf{false} otherwise.

\subsection{Description of the algorithm \textsf{AuxReach}}

First we construct a $h^{1-\beta}$-$\psc$ $C$ of $H$, using \expref{Theorem}{thm:constructpsc}. We also ensure that $x$ and $y$ are part of $C$ (if not then we add them). Let $I_1, I_2, \ldots, I_{\ell}$ be the connected components of $\sep{H}{C}$.

We maintain an array called \textsf{visited} of size $|C|$ to mark vertices or edges of the $\psc$ $C$. Each cell of \textsf{visited} corresponds to a distinct vertex or edge of $C$. For a vertex $v$ in $C$, we set $\textsf{visited}[v] := 1$ if there is a path from $x$ to $v$ in $H$, else it is set to 0. For an edge $e = (u,v)$ in $C$, we set $\textsf{visited}[e] := u'$ if (i) there is an edge $f = (u',v')$ that crosses $e$, (ii) there is a path from $x$ to $u'$ in $H$ and (iii) $f$ is the closest such edge to $u$. Else $\textsf{visited}[e]$ is set to \textsf{NULL}. Initially, for all vertex $v \in C$, $\textsf{visited}[v] := 0$ and for all edges $e \in C$, $\textsf{visited}[e] := \textsf{NULL}$. We say that a vertex $v$ is \emph{marked} if either $\textsf{visited}[v] = 1$ or $\textsf{visited}[e] = v$ for some edge $e$.

First set $\textsf{visited}[x]:=1$. We then perform an outer loop with $h$ iterations and in each iteration update certain entries of the array $\textsf{visited}$ as follows. For every vertex $v \in C$, the algorithm sets $\textsf{visited}[v]:=1$ if there is a path from a marked vertex to $v$ such that the internal vertices of that path all belong to only one component $I_i$. Similarly, for each edge $e = (u,v)$ of $C$, the algorithm sets $\textsf{visited}[e]:=u'$ if (i) there exists an edge $f = (u',v')$ which crosses $e$, (ii) there is a path from a marked vertex to $u'$ such that the internal vertices of that path all belong to only one component $I_i$ and, (iii) $f$ is the closest such edge to $u$. Finally we output \textsf{true} if $\textsf{visited}[y] = 1$ else output \textsf{false}. We use the procedure \textsf{AuxReach} recursively to check if there is a path between two vertices in a single connected component of $\sep{H}{C}$. A formal description of \textsf{AuxReach} is given in \expref{Algorithm}{alg:psgreach}.
\subsection{Proof of correctness of \textsf{AuxReach}}

Let $P$ be a path from $x$ to $y$ in $H$. Suppose $P$ passes through the components $I_{\sigma_1}, I_{\sigma_2}, \ldots, I_{\sigma_L}$ in this order. The length of this sequence %
is %
 at most $\lvert H \rvert$. As the path leaves the component $I_{\sigma_j}$ and goes into $I_{\sigma_{j+1}}$, it can do in the following two ways only:
\begin{enumerate}
\item The path exits $I_{\sigma_j}$ through a vertex $w$ of $\psc$ as shown in \expref{Figure}{fig:viaseparator}. In this case, \expref{Algorithm}{alg:psgreach} %
marks  %
the vertex $w$.
\item The path exits $I_{\sigma_j}$ through an edge $(u, v)$ whose other endpoint is in $I_{\sigma_{j+1}}$. By \expref{Lemma}{lem:closer}, this edge will cross an edge $e = (x',y')$ of the $\psc$. In this case, \expref{Algorithm}{alg:psgreach} %
marks  %
the vertex $u'$, such that there is an edge $(u', v')$ that crosses $e$ as well and $(u',v')$ is closer than $(u,v)$ to $x'$ and there is a path in $I_{\sigma_j}$ from a marked vertex to $u'$. By \expref{Lemma}{lem:closer}, the edge $(u', v)$ %
is  %
 in $H$ as well.
\end{enumerate}
\begin{algorithm}
\LinesNumbered
\caption{$\auxreach{H}{s}{t}$}
\label{alg:psgreach}
\KwIn{%
An %
induced subgraph $H$ of $\aux{G}$ and two vertices $x$ and $y$ in $H$ (let $G$ be an $m \times m$ grid %
digraph %
 and $h = |V(H)|$)}
\KwOut{\textsf{true} if there is a path from $x$ to $y$ in $H$ and \textsf{false} otherwise}

\lIf(\tcc*[f]{$m$ is a global variable where $G$ is an $m \times m$ grid %
digraph %
}){$h \leq m^{1/8} $}{
Use DFS to solve the problem}  \label{line:line1}
\Else{
Compute a $h^{1-\beta}$-$\psc$ $C$ of $H$ using \expref{Theorem}{thm:constructpsc}\; \label{line:calcpsep}
$C \gets C \cup \{x,y\}$\;
\lForEach{edge $e$ in $C$}{$\textsf{visited}[e] \gets \textsf{NULL}$}
\lForEach{vertex $v$ in $C$}{$\textsf{visited}[v] \gets 0$}
$\textsf{visited}[x] \gets 1$\;

\For{$i = 1$ to $\lvert H \rvert$}{ \label{alg:forloop}
\ForEach{edge $e = (u,v) \in C$}{

$\textsf{closestedge} \gets \textsf{NULL}$\;
\ForEach{marked vertex $w$}{
	\ForEach{$U \in \concom{\sep{H}{C}}$}{
		\ForEach{edge $f = (u',v')$ such that $f$ crosses $e$}{
			\If{$\auxreach{H[U \cup \{w, u' \}]}{w}{u'} = \textsf{true}$}{
				\If{$\textsf{closestedge} =\textsf{NULL}$  \textbf{or} $f$ is closer to $e$ than $\textsf{closestedge}$}{
					$\textsf{visited}[e] \gets u'$\;\label{line:recur1}
					$\textsf{closestedge} \gets f$\;
				}
			}
		}
	}
}
}
\ForEach{vertex $v \in C$}{
\If{$\exists w \exists U$ such that $w $ is 
a marked vertex\ {\bf and}\ $U \in \concom{\sep{H}{C}}$\
{\bf and}\ $\auxreach{H[U \cup \{w, v \}]}{w}{v} = \textsf{true}$))}{  
$\textsf{visited}[v] \gets 1$\; \label{line:recur2}}
}
}
\lIf{$\textsf{visited}[y] =1$}{\Return \textsf{true}}
\lElse{\Return \textsf{false}}
}
\end{algorithm}
Thus after the $j$-th iteration, \textsf{AuxReach}
traverses  %
the fragment of the path in the component $I_{\sigma_j}$ and either
marks  %
its endpoint or a vertex which is closer to the edge $e$ of $C$ which the path crosses. Finally, $y$ %
is  %
marked after $L$ iterations if and only if there is a path from $x$ to $y$ in $H$. We give a formal proof of correctness in \expref{Lemma}{lem:algo_correctness}. For a path $P = (u_1, u_2, \ldots u_t)$, we define $\tail{P} := u_1$ and $\head{P} := u_t$.

\begin{lemma}
\label{lem:algo_correctness}
Let $G$ be a grid %
digraph %
 and $H$ be %
an%
induced subgraph of $\aux{G}$. Then for any two vertices $x,y$ in $H$, there is a path from $x$ to $y$ in $H$ if and only if $\auxreach{H}{x}{y}$ returns \textsf{true}.
\end{lemma}
\begin{proof}
Firstly observe that, if a vertex is marked, then there is a path from some other marked vertex to that vertex in $H$. Hence if there is no path from $x$ to $y$ then $y$ is never marked by \textsf{AuxReach} and hence \textsf{AuxReach} returns \textsf{false}.

Now let $P$ be a path from $x$ to $y$ in $H$. We divide the path into subpaths $P_1, P_2 , \ldots , P_{\ell}$, such that for each $i$, all vertices of $P_i$ belong to $U \cup V(C)$ for some connected component $U$ in $\concom{\sep{H}{C}}$ and either (i) $\head{P_i} = \tail{P_{i+1}}$, or (ii) $e_i = (\head{P_i} , \tail{P_{i+1}})$ is an edge that crosses some edge $f_i \in C$. By \expref{Definition}{def:psc}, we have that if condition (i) is true then $\head{P_i}$ is a vertex in $C$, and if condition (ii) is true then $\head{P_i}$ and $\tail{P_{i+1}}$ belong to two different components of $\concom{\sep{H}{C}}$ and $e_i$ is the edge between them.

We claim that after $i$-th iteration of loop in \expref{Line}{alg:forloop} of \expref{Algorithm}{alg:psgreach}, either of the following two statements hold:
\begin{enumerate}
\item $\head{P_i}$ is a vertex in $C$ and $\textsf{visited}[\head{P_i}] = 1$.
\item There exists an edge $f_i = (u_i, v_i)$ of $C$ such that the edge $e_i = (\head{P_i}, \tail{P_{i+1}})$ crosses $f_i$ and there is an edge $g_i = (u_i',v_i')$ which crosses $f_i$ as well, such that $g_i$ is closer to $u_i$ than $e_i$ and $\textsf{visited}[f_i] = u_i'$.
\end{enumerate}
We prove the claim by induction. The base case holds since $x$ is marked at the beginning. We assume that the claim is true after the $(i-1)$-th iteration. We have that $P_i$ belongs to $U \cup V(C)$ for some connected component $U$ in $\concom{\sep{H}{C}}$.
\begin{description}
\item[Case 1 ($\head{P_{i-1}} = \tail{P_{i}} = w$(say)):] By induction hypothesis $w$ was marked after the $(i-1)$-th iteration. If $\head{P_i}$ is a vertex in $C$ then it will be marked after the $i$-th iteration in \expref{Line}{line:recur2}. On the other hand if $e_i = (\head{P_i} , \tail{P_{i+1}})$ is an edge that crosses some edge $f_i = (u_i,v_i) \in C$ then in the $i$-th iteration in \expref{Line}{line:recur1}, the algorithm marks a vertex $u_i'$ such that, $g_i = (u_i',v_i')$ is the closest edge to $u_i$ that crosses $f_i$ and there is a path from $w$ to $u_i'$.

\item[Case 2 ($e_{i-1} = (\head{P_{i-1}} , \tail{P_{i}})$ crosses some edge $f_{i-1} = (u_{i-1},v_{i-1}) \in C$):] By induction hypothesis, there is an edge $g_{i-1} = (u_{i-1}',v_{i-1}')$ which crosses $f_{i-1}$ as well, such that $g_{i-1}$ is closer to $u_{i-1}$ than $e_{i-1}$ and $\textsf{visited}[f_{i-1}] = u_{i-1}'$. By \expref{Lemma}{lem:closer} there is an edge in $H$ between $u_{i-1}'$ and $\tail{P_{i}}$ as well.
Now if $\head{P_i}$ is a vertex in $C$ then it will be marked after the $i$-th iteration in \expref{Line}{line:recur2} by querying the %
digraph %
 $H[U \cup \{u_{i-1}', \head{P_i}\}]$. On the other hand if $e_i = (\head{P_i} , \tail{P_{i+1}})$ is an edge that crosses some edge $f_i = (u_i,v_i) \in C$ then in the $i$-th iteration in \expref{Line}{line:recur1}, \textsf{AuxReach} queries the %
digraph %
 $H[U \cup \{u_{i-1}', u_i'\}]$ and marks a vertex $u_i'$ such that, $g_i = (u_i',v_i')$ is the closest edge to $u_i$ that crosses $f_i$ and there is a path from $u_{i-1}'$ to $u_i'$.\qedhere\end{description}\end{proof}

Our subroutine %
solves  %
reachability in a subgraph $H$ (having size $h$) of $\aux{G}$. We do not explicitly store a component of $\concom{\sep{H}{C}}$, since it might be too large. Instead, we identify a component with the lowest indexed vertex %
 in it and use Reingold's algorithm on $\sep{H}{C}$ to determine if a vertex is %
 in that component. We require $\widetilde{O}(h^{1/2 + \beta/2})$ space to compute a $h^{1-\beta}$-$\psc$ by \expref{Theorem}{thm:constructpsc}. We can potentially mark all the vertices of the $\psc$ and for each edge of
the  %
$\psc$ we mark at most one additional vertex. Since the size of
the  %
$\psc$ is at most $O(h^{1/2 + \beta/2})$, we require $\widetilde{O}(h^{1/2 + \beta/2})$ space. The algorithm recurses on a %
digraph %
 with $h^{1 - \beta}$ vertices.  %
Since we stop the recursion when $h \le m^{1/8}$
(\expref{Line}{line:line1} of \expref{Algorithm}{alg:psgreach}), %
the depth of the recursion is at most
$\lceil -3/\log(1-\beta)\rceil$,  %
which is a constant.

Since the %
digraph %
 $H$ is given implicitly in our algorithm, an additional polynomial overhead is involved in obtaining its vertices and edges. However, the total time complexity %
remains  %
a polynomial in the number of vertices since the recursion depth is constant.

\begin{lemma}
Let $G$ be an $m \times m$ grid %
digraph %
 and $H$ be%
an%
induced subgraph of $\aux{G}$ with $h$ vertices. For every $\beta > 0$, \textsf{AuxReach} runs in $\widetilde{O}({h}^{1/2 + \beta/2})$ space and polynomial time.
\end{lemma}
\begin{proof}
  Since the size of a component $U$ in $\concom{\sep{H}{C}}$ might be too large, we will not explicitly store it. Instead we identify a component by the lowest index vertex %
 in it and use Reingold's algorithm on $\sep{H}{C}$ to determine if a vertex is %
 in $U$. Let $S(m,h)$ and $T(m,h)$ denote the space and time complexity functions respectively of \textsf{AuxReach}, where $G$ is an $m \times m$ grid %
digraph %
 and $h$ is the number of vertices in the %
digraph %
 $H$. As noted earlier the depth of the recursion is at most 
$d:= \lceil -3/\log(1-\beta)\rceil$. %

Consider $S(m,h)$ for any $h > m^{1/8}$. By \expref{Theorem}{thm:constructpsc}, we require $\widetilde{O}(h^{1/2 + \beta/2})$ space to execute \expref{Line}{line:calcpsep}. We can potentially mark all the vertices of $C$ and for each edge $e$ of $C$ we store at most one additional vertex in $\textsf{visited}[e]$. Since the size of $C$ is at most $O(h^{1/2 + \beta/2})$, we require $\widetilde{O}(h^{1/2 + \beta/2})$ space to store $C$. By induction, a call to \textsf{AuxReach} in \expref{line}{line:recur1} \expref{and}{line:recur2} requires $S(m, h^{1 - \beta})$ space which can be subsequently reused. Hence the space complexity satisfies the following recurrence. Then,
\[S(m,h) =
\begin{cases}
S(m, h^{1 - \beta}) + \widetilde{O}(h^{1/2 + \beta/2}) & h > m^{1/8}\\
\widetilde{O}(h) & \textrm{otherwise}.
\end{cases}\]
Solving we get $S(m,h) = \widetilde{O}(h^{1/2 + \beta/2} + m^{1/4})$.

Next we measure the time complexity of \textsf{AuxReach}. Consider the case when $h > m^{1/8}$. The total number of steps in \textsf{AuxReach} is some polynomial in $h$, say $p$. Moreover \textsf{AuxReach} makes $q$ calls to \textsf{AuxReach}, where $q$ is some other polynomial in $h$. Hence $q(h) \leq p(h)$. Then,
\[T(m,h) =
\begin{cases}
q\cdot T(h^{1 - \beta}) + p & h > m^{1/8} \\
O(h) & \textrm{otherwise}.
\end{cases} \]
Solving the above recurrence we get $T(m,h) = O(p\cdot q^d + m^{1/4}) = O(p^{2d} + m^{1/4})$.
\end{proof}

\section{Solving grid %
digraph %
}
\label{sec:final}
Let $G$ be an $m \times m$ grid %
digraph %
. As mentioned in the introduction, our objective is to run the \expref{algorithm}{alg:psgreach} on the %
digraph %
 $\aux{G}$. Consider two vertices $x$ and $y$ of $\aux{G}$. Note that, by definition, $y$ is reachable from $x$ in $\aux{G}$ if and only if $y$ is reachable from $x$ in $G$. Hence it is sufficient to work with the %
digraph %
 $\aux{G}$. However, we do not have explicit access to the edges of $\aux{G}$. Note that we can obtain the edges of $\aux{G}$ by solving the corresponding subgrid of $G$ to which that edge belongs. If the subgrid is small enough, then we use a standard linear %
space-traversal  %
algorithm. Otherwise, we use our algorithm recursively on the subgrid. \expref{Algorithm}{alg:gridreach} outlines this method.
\begin{algorithm}
\caption{\textsf{GridReach}$(\widehat{G}, \widehat{s}, \widehat{t}, m)$}
\label{alg:gridreach}
\KwIn{A grid %
digraph %
 $\widehat{G}$ and two vertices $\widehat{s}$, $\widehat{t}$ of $\widehat{G}$ and a positive integer $m$}
\KwOut{\textsf{true} if there is a path from $s$ to $t$ in $G$ and \textsf{false} otherwise}
\If{$\widehat G$ is smaller than $m^{1/8} \times m^{1/8}$ grid}{
Use Depth-First Search to solve the problem\;
}
\Else
{
Use \textsf{ImplicitAuxReach}$(\aux{G}, \widehat{s}, \widehat{t})$ to solve the problem\;
\tcc*[f]{\textsf{ImplicitAuxReach} executes the same way as \textsf{AuxReach} except for the case when it queries an edge $(u, v)$ in a block $B$ of $\aux{G}$. In this case, the query is answered by calling \textsf{GridReach}$(B, u, v, m)$ where $B$ is the subgrid in which edge $(u,v)$ might belong.}
}
\end{algorithm}

Consider an $\widehat{m} \times \widehat{m}$ grid %
digraph %
 $\widehat{G}$. Let $S(\widehat{m})$ be the space complexity and $T(\widehat{m})$ be the time complexity of executing \textsf{GridReach} on $\widehat{G}$. Note that the size of $\aux{\widehat{G}}$ is at most $\widehat{m}^{1 + \alpha}$. For $\widehat{m} > m^{1/8}$, the space required to solve the grid %
digraph %

is   %
$S(\widehat{m}) = S(\widehat{m}^{1-\alpha}) + \widetilde{O}((\widehat{m}^{1+\alpha})^{1/2 + \beta/2})$. This is because, a query whether $(u,v) \in \widehat{G}$%
invokes  %
a recursion which %
requires  %
$S(\widehat{m}^{1-\alpha})$ space and the main computation of \textsf{ImplicitAuxReach} %
can  %
be done using $\widetilde{O}((\widehat{m}^{1+\alpha})^{1/2 + \beta/2})$ space. Hence we get the following recurrence for space complexity.
\[S(\widehat{m}) =
\begin{cases}
S(\widehat{m}^{1 - \alpha}) + \widetilde{O}((\widehat{m}^{1+\alpha})^{1/2 + \beta/2}) & \widehat{m} > m^{1/8} \\
\widetilde{O}(\widehat{m}^{1/4}) & \textrm{otherwise}
\end{cases}\]
Similar to the analysis of \textsf{AuxReach}, for appropriate polynomials $p$ and $q$, the time complexity %
satisfies  %
the following recurrence:
\[T(\widehat{m}) = \begin{cases}
q (\widehat{m}) \cdot T(\widehat{m}^{1 - \alpha}) + p(\widehat{m}) & \widehat{m} > m^{1/8} \\
O(\widehat{m}) & \textrm{otherwise}.
\end{cases}\]
Solving we get $S(m) = \widetilde{O}(m^{1/2 + \beta/2 + \alpha/2 + \alpha\beta/2})$ and $T(m) = \poly(m)$. For any constant $\epsilon>0$, we can chose $\alpha$ and $\beta$ such that $S(m) = {O}(m^{1/2 + \epsilon})$.

\section{Conclusion}
Our result improves upon the known time--space bounds on the grid %
digraph %
. Asano et al. \cite{Asano} used a clever idea of exploiting the grid structure of the input %
digraph %
to reduce its size. Their exploitation destroyed the grid structure of the input
digraph %
, but they did not go far enough to destroy its planar structure as well. Our exploitation goes a step further. We get a non-planar auxiliary %
digraph %
 that is smaller as a result and has just enough structure to solve reachability in a space-efficient manner. It remains to be seen if the structure could be further exploited to get a smaller \emph{auxiliary-like} %
digraph %
 in which reachability can be solved.

It is known that reachability in planar %
digraph %
s %
can  %
be reduced to reachability in grid %
digraph %
s in logarithmic space \cite{Allender}. However, such a reduction results in at least a quadratic blow-up in the size of the %
digraph %
. In principle, it would seem that an improvement to the 
state of the art  %
for the planar %
digraph %
 can be obtained by improving the result for grid %
digraph %
s. However, we feel that this would be a difficult direction to go. Note that we solve an $h$-vertex auxiliary %
digraph %
 in $O(h^{1/2 + \beta/2})$ space as a subroutine for our grid%
digraph %
 algorithm. A planar %
digraph %
 with its embedding can be thought of as an auxiliary %
digraph %
, and hence our algorithm contains within itself a solution to the planar %
digraph %
 as well. For this reason, we feel that directly solving a planar %
digraph %
would be easier than going down the %
grid digraph %
 routine. A significantly different approach would be required to directly design an algorithm for %
grid digraph %
, one which does not use a solution for the class of planar %
digraph %
s, or its superclass, as a subroutine.

\bibliographystyle{tocplain}   %

\bibliography{bibstrings,v019a002,bibtail}

\begin{tocauthors}
\begin{tocinfo}[Jain]
Rahul Jain \orcidlink{0000-0002-8567-9475}\\
Research associate\\
Department of Theoretical Computer Science\\
Fernuniversit\"at in Hagen\\
Germany\\
rahul\tocdot{}jain\tocat{}fernuni-hagen\tocdot{}de\\   %
\url{https://www.fernuni-hagen.de/ti/en/team/rahul.jain}   %
\end{tocinfo}
\begin{tocinfo}[Tewari]
  Raghunath Tewari\\
 Associate professor\\
 Department of Computer Science and Engineering \\
 Indian Institute of Technology Kanpur\\
  rtewari\tocat{}cse\tocdot{}iit\tocdot{}ac\tocdot{}in \\
  \url{https://www.cse.iitk.ac.in/users/rtewari}
\end{tocinfo}
\end{tocauthors}

\begin{tocaboutauthors}
\begin{tocabout}[Jain]
  \textsc{Rahul Jain} is a postdoctoral research associate at Fernuniversit\"at in Hagen, Germany. He obtained his doctorate at the Indian Institute of Technology Kanpur, India in 2020 under \href{https://www.cse.iitk.ac.in/users/rtewari/}{Dr. Raghunath Tewari}.  %
\end{tocabout}

\begin{tocabout}[Tewari]
 \textsc{Raghunath Tewari} is an Associate Professor in the Computer Science and Engineering %
Department  %
at the Indian Institute of Technology Kanpur. He obtained his undergraduate degree in mathematics and computer science from Chennai Mathematical Institute in 2005.
Thereafter he did his Masters in 2007 and \phd\ in 2011 in computational complexity theory
at  %
the University of Nebraska -- Lincoln under \href{http://cse.unl.edu/~vinod/}{Dr. N. V. Vinodchandran}.  %
He is interested in computational complexity theory, algorithms and graph theory.

\end{tocabout}
\end{tocaboutauthors}

\end{document}